\documentclass[runningheads]{llncs}
\usepackage{amssymb,amsfonts,amsmath}
\usepackage{booktabs,tabularx}
\usepackage{url,xspace,soul,wrapfig}
\usepackage{graphicx,hyperref,paralist}
\usepackage[font=small, hypcap=true]{subfig,caption}
\usepackage{paralist,enumerate}
\usepackage[textsize=footnotesize]{todonotes}
\usepackage{enumerate,todonotes}
\usepackage{soul,lineno}
\usepackage{xcolor}
\usepackage[utf8]{inputenc}
\usepackage[capitalise]{cleveref}
\newcommand{\old}[1]{{}}
\newtheorem{obs}[lemma]{Observation}
\newcommand{\np}{\mathsf{NP}}
\usepackage{mdframed}
\usepackage{pifont}
\DeclareMathOperator{\poly}{poly}
\DeclareMathOperator*{\col}{col}

\nolinenumbers

\hypersetup{
    colorlinks=true,
    linkcolor=blue,
    filecolor=violet,  
    citecolor=violet,
    urlcolor=cyan,
    pdftitle={Overleaf Example},
    pdfpagemode=FullScreen,
}

\newenvironment{sproof}{%
  \proof}{\endproof}
\usepackage{lineno}

\linenumbers
\title{Geometric Planar Networks\\ on Bichromatic Collinear Points\thanks{A preliminary version of this paper appeared in
the proceedings of 6th Annual International Conference on Algorithms and Discrete Applied Mathematics (CALDAM~2020), and the full version appeared in the Theoretical Computer Science (TCS) journal.
The research of Sayan Bandyapadhyay is partly funded by the European Research Council (ERC) via grant LOPPRE, reference 819416. Research of Sujoy Bhore and Martin N\"{o}llenburg is supported by the Austrian Science Fund (FWF) grant P 31119.}}

\author{Sayan Bandyapadhyay\inst {1}\and
Aritra Banik\inst{2}\and 
Sujoy Bhore\inst{3}\and
Martin N\"{o}llenburg\inst{4}}
\authorrunning{S. Bandyapadhyay, A. Banik, S. Bhore, M. N\"{o}llenburg}

\institute{Department of Informatics, University of Bergen, Norway\\ \email{sayangcelt@gmail.com}
\and
School of Computer Sciences, NISER, Bhubaneswar, India\\ \email{aritrabanik@gmail.com}
\and
Indian Institute of Science Education and Research, Bhopal, India\\
\email{sujoy.bhore@gmail.com}
Algorithms and Complexity Group, Technische Universit\"{a}t Wien, Austria\\
\and 
\email{noellenburg@ac.tuwien.ac.at}
}

\begin{document}

\nolinenumbers

\maketitle
\begin{abstract}
We study three classical graph problems -- Hamiltonian path, minimum spanning tree, and minimum perfect matching on geometric graphs induced by bichromatic (red and blue) points. These problems have been widely studied for points in the Euclidean plane, and many of them are $\np$-hard. 
In this work, we consider these problems for collinear points. We show that almost all of these problems can be solved in linear time in this setting. 
\end{abstract}


\section{Introduction}
\label{sec:introduction}

In this article, we study three classical graph problems on geometric graphs induced by bichromatic (red and blue) points. Suppose, we are given a set $R$ of $n$ red points and a set $B$ of $m$ blue points in the Euclidean plane. Consider the complete bipartite graph $G(R,B,E)$ on $R \cup B$, where the set $E$ of edges contains all bichromatic edges between the red points and the blue points. Also, suppose the graph $G(R,B,E)$ is embedded in the plane: the points are the vertices and each edge is represented by the line segment between the two corresponding endpoints. We denote these edges as \emph{bichromatic segments}, where each bichromatic segment connects a red point with a blue point. 
A subgraph $G'$ of $G(R,B,E)$ (or equivalently a subset of edges of $E$) is called \emph{non-crossing} (or planar) if no pair of the edges of $G'$ cross each other. Next, we discuss the three graph problems on the bipartite graph $G(R,B,E)$ 
that we study in this paper. 

In the \textsc{Bichromatic Hamiltonian path} problem, the objective is to find a path in $G(R,B,E)$ that spans all the red and blue points. Equivalently, one would like to find a polygonal chain that connects all the red points and the blue points alternately through bichromatic segments. It is not hard to see that a Hamiltonian path exists in $G(R,B,E)$ if and only if $m-1\le n\le m+1$, and if there exists one, it can be computed efficiently, as $G(R,B,E)$ is a complete bipartite graph. A more interesting problem is the \textsc{Non-crossing bichromatic Hamiltonian path} problem where the objective is to find a non-crossing Hamiltonian path. Note that one can construct instances with $m-1\le n\le m+1$, for which it is not possible to find any non-crossing Hamiltonian path. In Figure~\ref{fig:nohampath}, 
we demonstrate two such instances.\footnote{In all the figures throughout the paper, 
we show red (resp. blue) points by squares (resp. disks).}
Figure~\ref{fig:nohampath}(a) has eight points, where four of them lie on a horizontal line $L$ and the remaining four lie on a line parallel to $L$. Notice that there must be one red and one blue point with degree~$1$. One can verify by enumerating all possible paths that there is no non-crossing Hamiltonian path that spans these points. 
The example in Figure~\ref{fig:nohampath}(b) has thirteen points in general position, i.e., no three points are collinear, and it also does not admit a non-crossing Hamiltonian path. Indeed, if $2\le n\le 12$, then for any given $\lfloor{n/2}\rfloor$ red (resp. blue) points and $\lceil{n/2}\rceil$ blue (resp. red) points in general position, there exists a non-crossing Hamiltonian path~\cite{kaneko1}. Due to the uncertainty of the existence of non-crossing Hamiltonian paths in the general case, researchers have also considered the problem of finding a non-crossing alternating path of length as large as possible~\cite{KanekoK99,KynclPT08}. Recently, 
Mulzer and Valtr~\cite{MulzerV20} showed that for a set of points in convex position, there is an absolute constant $\varepsilon>0$, independent of $n$ and of the bicoloring of the points, such that $P$ always admits a non-crossing alternating path of length at least $(1+\varepsilon) n$. 
Garcia and Tejel~\cite{garcia2017polynomially} designed polynomial time algorithms for special geometric instances.

\begin{figure}[tbp]
\centering
\includegraphics[width=0.55\textwidth]{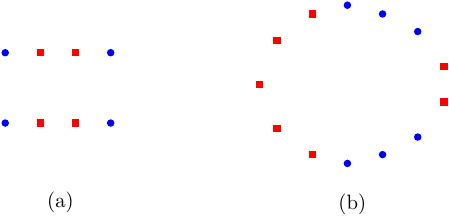}
\caption{Two instances without a non-crossing Hamiltonian path. Figure~(a) is not in general position, Figure~(b) (borrowed from \cite{kaneko2003discrete}) is in general position.}
\label{fig:nohampath}
\end{figure}

\old{A related problem is the \textsc{Bichromatic traveling salesman path} \textsc{(Bichromatic TSP)} problem where one would like to find a minimum weight Hamiltonian path in $G(R,B,E)$. The weight of each edge is the (Euclidean) length of the corresponding segment. The weight of a path is the sum of the weights of the edges along the path. For simplicity,  we assume $n=m$. A straightforward reduction from the (monochromatic) Euclidean TSP~\cite{Arora98} (replace each point by a bichromatic pair that is small distance apart) shows that Bichromatic TSP is also $\np$-hard. One simple, but powerful fact is that an optimum Euclidean TSP is always non-crossing. This helps to obtain a PTAS~\cite{Arora98} for the problem. However, an optimum Bichromatic TSP is not necessarily non-crossing which makes its computation much harder compared to Euclidean TSP. 
The best known approximation factor for the Bichromatic TSP problem is $2$ due to 
Frank~{\em et al.}~\cite{frank1998bipartite} who improved the $2.5$-approximation of Anily~{\em et al.}~\cite{anily1992swapping}.
For a set of collinear points, Evans~{\em et al.}~\cite{EvansLMW16} gave a quadratic time algorithm for computing an optimum non-crossing TSP, every edge of which is a poly-line with at most two bends.  
}


Next, we consider the \textsc{Bichromatic spanning tree} problem where the objective is to compute a minimum weight spanning tree
of $G(R,B,E)$, 
where the weight of the edge between each pair of points is the Euclidean distance between those two points. Note that this problem can be solved efficiently by any standard minimum spanning tree algorithm. Biniaz~{\em et al.}~\cite{biniaz2018spanning} showed that minimum bichromatic spanning tree can be computed in $O(n\log n)$ time. 
Moreover, their algorithm extends to the multicolored version with the time complexity $O(n\log n\log k)$ where $k$ is the number of different colors (or the size of the multipartition in a complete multipartite geometric graph).
A more interesting problem is the \textsc{Non-crossing bichromatic spanning tree} problem where additionally the computed tree must be  non-crossing.
Borgelt~{\em et al.}~\cite{BorgeltKLLMSV09} showed that this problem is $\np$-hard. For points in general position, they gave a near-linear time $O(\sqrt{n})$-approximation. On the other hand, for points in convex position, they gave an exact cubic-time algorithm. 
Another line of work that received much attention is where the task is to find a degree-bounded non-crossing spanning tree \cite{BiniazBMS18}. 

Finally, we consider the \textsc{Bichromatic matching} problem. Again assume that $n=m$ for simplicity. We would like to find a minimum weight perfect matching in $G(R,B,E)$. 
The weight of an edge is the Euclidean distance between its corresponding points.
It is a well-known fact that a minimum weight bichromatic matching (for points in general position) in the plane is always non-crossing, which follows from the observation that the sum of the diagonals of a convex quadrilateral is strictly larger than the sum of any pair of opposite sides. 
This implies that using any standard bipartite matching algorithm one can solve a \textsc{Non-crossing Bichromatic matching} exactly. But, algorithms with better running time have been designed by exploiting the underlying geometry of the plane. Recently, Kaplan~{\em et al.}~\cite{KaplanMRSS17} designed an $O(n^2 \poly(\log n))$ algorithm for the problem improving the $O(n^{2+\epsilon})$ algorithm due to Agarwal~{\em et al.}~\cite{AgarwalES99}, where $\poly(.)$ is a polynomial function. In~\cite{MulzerV20}, Mulzer and Valtr showed that for a set of in convex position, there always exists a non-crossing bichromatic separated matching\footnote{This is a properly colored matching whose segments are pairwise disjoint and intersected by common line.} on at least $(1+\varepsilon) n$ points of $P$, where $\varepsilon>0$ that is independent of $n$ 
and of the bicoloring of the points.
Abu-Affash \emph{et al.}~\cite{Abu-AffashBCMS15} studied the bottleneck variant of the non-crossing matching problem. This problem is known to be $\np$-Hard. They gave an $O(n\log n)$-time approximation algorithm which computes a non-crossing matching of size at least $\frac{2n}{5}$ edges, 
whose edges have length at most $\sqrt{2}+\sqrt{3}$ times the bottleneck. Biniaz \emph{et al.}~\cite{BiniazBMS16} studied the 
non-crossing geodesic spanning trees, Hamiltonian cycle, and perfect matching in a simple polygon.


In this article, we consider the above mentioned problems for collinear points on a real line\footnote{Throughout the paper we assume that all points are distinct.}. We assume that the points are given in their sorted order. 
We note that the case of non-crossing graphs on collinear points is closely related to 
1-page or 2-page book embeddings~\cite{bk-btg-79}, which have all vertices placed on a line (called the spine) and the edges drawn without crossings in one or two of the halfplanes (but not in both) defined by the spine (called the pages). In our case we assume that the edges are drawn as (circular) arcs or 1-bend polylines either above or below the spine.\footnote{For the sake of convenience, we draw edges as simple curves in all the figures.} We assume that their  weight is  given by the Euclidean distance of their endpoints. If the arcs are drawn infinitesimally close to the spine, these weights correspond to the lengths of the arcs.

\subsection{Our Results}
The main results obtained in this work are the following. 

\vspace{-.2cm}
\begin{enumerate}
    \item[\ding{229}] \textsc{Non-crossing Hamiltonian path for collinear points} -- We prove that for any collinear configuration of the points with $|R|=|B|$, there always exists a non-crossing Hamiltonian path. Additionally, we give a linear-time algorithm for computing such a path (Section \ref{sec:hampath}). As noted before, if the points are not collinear and lie in the plane, the existential claim is not necessarily true. 
 
 \vspace{.1cm}   
    
    \item[\ding{229}] \textsc{Minimum spanning tree for collinear points} --  We give a linear-time algorithm for computing a minimum-weight spanning tree, and a quadratic-time algorithm for computing a minimum-weight non-crossing spanning tree for collinear points and edges on a single page (Section \ref{sec:spantree}). In contrast, the problem in the non-crossing case is $\np$-Hard in the plane. 
   
 \vspace{.1cm}  
    \item[\ding{229}] \textsc{Minimum non-crossing matching for collinear points} -- We give a linear-time algorithm that computes a minimum-weight non-crossing perfect matching for collinear points and edges on a single page (Section \ref{sec:matching}). We note that the most efficient algorithm for this problem in the plane runs in $O(n^2 \poly(\log n))$ time. 
\end{enumerate}
We note that even in this simple one-dimensional case these problems become sufficiently challenging if one is constrained to use only linear (or near-linear) time. Our study of one-dimensional case is partly motivated due to the intractability of the problems or lack of efficient algorithms in the plane. Note that for all the problems, we have obtained improved or more interesting results compared to the planar case. Our findings  give a better understanding of these problems and our work can be considered as a stepping stone towards achieving improved results in the plane.   

\paragraph{Remark.} In a parallel work, Aichholzer et al.~\cite{aichholzer2019hamiltonian} obtained a different linear time algorithm for \textsc{Non-crossing Hamiltonian path for collinear points}.


\subsection{Related Work.}  A closely related problem to \textsc{Bichromatic Hamiltonian path} is the \textsc{Bichromatic traveling salesman problem} \textsc{(Bichromatic TSP)} problem, where one would like to find a minimum-weight Hamiltonian cycle in $G(R,B,E)$. The weight of each edge is the (Euclidean) length of the corresponding segment. The weight of a path is the sum of the weights of the edges along the path. 
We assume $n=m$, otherwise, 
there is no Bichromatic Hamiltonian cycle.
A straightforward reduction from the (monochromatic) \textsc{Euclidean TSP}~\cite{Arora98} (replace each point by a bichromatic pair that is small distance apart) shows that \textsc{Bichromatic TSP} is also $\np$-hard. One simple, but powerful fact is that an optimum Euclidean TSP is always non-crossing. This helps to obtain a PTAS~\cite{Arora98} for the problem. However, an optimum Bichromatic TSP is not necessarily non-crossing which makes its computation much harder compared to Euclidean TSP. 
The best known approximation factor for the Bichromatic TSP problem is $2$ due to 
Frank~{\em et al.}~\cite{frank1998bipartite} who improved the $2.5$-approximation of Anily~{\em et al.}~\cite{anily1992swapping}.
For a set of collinear points, Evans~{\em et al.}~\cite{EvansLMW16} gave a quadratic time algorithm for computing an optimum non-crossing TSP, every edge of which is a poly-line with at most two bends. We note that, in the existing research literature, the framework used by  Evans~{\em et al.}~\cite{EvansLMW16} is the closest to our work.  

Colannino et al.~\cite{ColanninoDHLMRST07} gave a near-linear time algorithm for the many-to-many matching problem for red-blue points on a line, which is similar to bichromatic matching. In many-to-many matching, we are given a bipartite graph, and we want to find a minimum-weight subset of edges that span all the vertices. Note that in contrast to our problem, here 
two edges can share an endpoint. 




\section{Non-crossing Hamiltonian Path for Collinear Points}
\label{sec:hampath}

If we would require for the collinear point set that each edge of a Hamiltonian path is a straight-line segment, the problem becomes trivial: an input instance can have a non-crossing Hamiltonian path if and only if the colors of the points alternate. Therefore, we consider the case where edges are represented by circular arcs drawn in the halfplane either above or below the line.

\begin{definition}
\textbf{Non-crossing Hamiltonian path for collinear points}. Given a set $R$ of $n$ red points and a set $B$ of $n$ blue points on a line, find a non-crossing geometric path $\pi$ in the plane such that $\pi$ consists of a sequence of circular arcs above or below $H$, each of which connects a red and a blue point and $\pi$ spans all the input points. 
\end{definition}

Note that in the above definition if the path is allowed to use arcs only from above (resp. below) $H$, then there might not exist such a Hamiltonian path. For example, consider the configuration in Figure \ref{fig:nohampathline}, which contains 16 points. Also, consider the continuous monochromatic chunks of points. Thus, the first chunk contains four blue points, the second chunk contains two red points, and so on. Note that the first blue point must be connected to the last red point in any non-crossing Hamiltonian path. Otherwise, there are two options: it gets connected to (i) a different point of the fourth chunk and (ii) a point of the second chunk. In the first case, the last red point cannot get connected to a blue point. In the second case, at least one point of the first chunk cannot get connected to a red point. Thus, the first blue point must be connected to the last red point. Moreover, for similar reasons, the second blue point must be connected to the second to last red point, the third blue point must be connected to the third to last red point and the fourth blue point must be connected to the fourth to last red point. Figure~\ref{fig:nohampathline} shows such an example Hamiltonian path. It is easy to verify that this path cannot be completed by connecting all the points. 


\begin{figure}[tbp]
\centering
\includegraphics[width=0.6\textwidth]{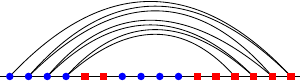}
\caption{Example of a set of collinear points for which a non-crossing Hamiltonian path does not exist if the arcs can be drawn only above the spine.}
\label{fig:nohampathline}
\end{figure}

In the light of the above discussion, we would like to find a non-crossing Hamiltonian path each of whose edges is a circular arc that lies either above or below $H$. 
First, we give a simple and intuitive construction of such a path for any configuration of points. The construction itself takes polynomial time, hence giving a polynomial-time algorithm for computation of such a path. Later, we give a more involved algorithm that runs in linear time. 


\subsection{The Construction}
\label{sec:hampathexistence}

To construct the path, we start with any bichromatic matching (not necessarily non-crossing) of the points. Note that each matching edge is a segment on $H$. We will connect these edges to obtain a Hamiltonian path. First, we form a hierarchical (or laminar) structure of these matching edges. Informally, the matching edges are hierarchical if any two edges are either separated or one is nested in the other. 

\begin{definition}
A set of matching
edges $M$ are hierarchical (or laminar) if for any two edges $(u,v),(w,x) \in M$ with $u < v$, $w < x$ and $u < w$, either $u < v < w < x$ $((u,v),(w,x)$ are separated$)$ 
or $u < w< x< v$ $((w,x)$ is nested in $(u,v))$.
\end{definition}

\begin{figure}[ht!]
 \centering
 \includegraphics[width=0.6\textwidth]{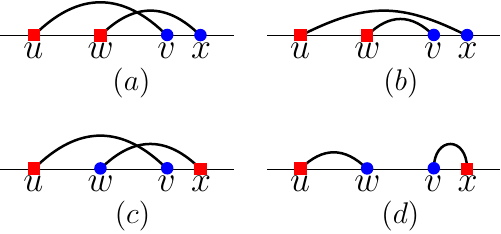}
  \caption{Illustration showing uncrossing of edges in the first ((a),(b)) and second subcases ((c),(d)). (a) and (c) showing 
  edges before uncrossing. (b) and (d) showing 
  edges after uncrossing.}
\label{fig:uncross}
\end{figure}

Given any matching $M$ for $R\cup B$, we can change it to a hierarchical matching in the following way. If there are two edges $(u,v),(w,x) \in M$ with $u < v$, $w < x$, $u < w$ that are not separated and none of them is nested in the other, then it must be the case that $u < w < v < x$. Now, there are two subcases, depending on the colors of $u$ and $w$. If $u,w$ are red or $u,w$ are blue, we replace the edges $(u,v),(w,x)$ by the two bichromatic edges $(u,x),(w,v)$. Otherwise, either $u$ is red, $w$ is blue or $u$ is blue, $w$ is red. In that case, we replace the edges $(u,v),(w,x)$ by the two bichromatic edges $(u,w),(v,x)$. See Figure \ref{fig:uncross} for a demonstration. Note that in all the cases, the new pair of edges does not violate the hierarchical structure. We repeat the process for each pair of edges that violates the condition. Newly formed edges might violate the condition with respect to other edges. However, if an edge is removed, it is never added back, and thus the process will eventually stop at some point when no pair of edges violates the condition any more. 

Next, we associate levels with each matching edge of $M$ in a recursive way. In the base case, for each edge that does not nest any other edge, set its level to 1. Now, suppose we have defined edges of level $j$ for each $j\le i-1$ for $i\ge 2$. An edge $(u,v)$ has level $i$, if it nests a level $i-1$ edge, and for any level $i-1$ edge $(w,x)$ that it nests, there is no other edge that is nested in $(u,v)$ that also nests $(w,x)$ (see Figure \ref{fig:fig1}). Note that the level of each edge is unique. Let $L$ be the maximum level. 

\begin{figure}[tbp]
 \centering
 \includegraphics{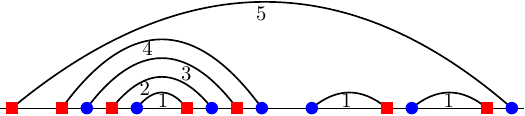}
  \caption{Illustration showing the levels of the edges.}
\label{fig:fig1}
\end{figure}

For any edge $(u,v)$ of $M$ with level $j$, call the points that lie between $u$ and $v$ including $u$ and $v$ as a level $j$ block. Thus, a level $l$ block is a union of disjoint blocks of levels at most $l-1$ and two special points which are the first and last point of the block. 

\begin{obs}
Each block as defined above contains the same number of red and blue points. 
\end{obs}

\begin{proof}
We prove this by induction on the level $l$ of the blocks. In the base case, $l=1$ and each level 1 block consists of the two endpoints of the corresponding edge. So, the statement holds in this case. Now, suppose the statement is true for all blocks with level $l\le j$. Consider any block $\beta$ of level $j+1$ corresponding to the edge $(u,v)$. Now, let $\beta$ be the union of the disjoint blocks $\beta_1,\ldots,\beta_t$ along with the points $u,v$. As each $\beta_i$ has level at most $j$, by induction, it has the same number of red and blue points. As $(u,v)$ is bichromatic, it follows that $\beta$ as well has the same number of red and blue points.    \qed     
\end{proof}

We compute the Hamiltonian path for all the blocks in a bottom up manner. The path of a level 1 block is the matching edge itself which defines the block. Additionally, for each block, we compute a path for the block that satisfies the following two invariants.

\begin{itemize}
    \item The first point of the block is an endpoint of the path. 
    \item If an endpoint $p$ of the path is not an endpoint of the block, then the path cannot contain two edges $(u,v),(w,x)$ with $u < v$ and $w < x$, such that $(u,v)$ lies above $H$, $(w,x)$ lies below $H$, $u < p < v$, and $w< p < x$. 
\end{itemize}
Informally, the second condition states that the endpoint of the path that is not an endpoint of the block should be available for connecting with an edge at least from one side. Note that the paths for level 1 blocks trivially satisfy the invariants. Now, assume that we have computed the paths for all the level $j$ blocks for $j\le l-1$ and $l\ge 2$ that satisfy the invariants. We show how to compute the path for a level $l$ block $S$ that also satisfies the invariants. Let $u,v$ be the endpoints of the block. Also let $S_1,\ldots,S_t$ be the blocks, sorted w.r.t the index of the first point in increasing order, whose union with the set $\{u,v\}$ forms the block $S$. As $S_i$ has level at most $l-1$, we have already computed the path of $S_i$ for all $i$. We show, by induction, how to construct the path $T'$ for the points in $\cup_{j=1}^i S_j$ for all $2\le i\le t$. Then, we show how to join the edge $(u,v)$ with $T'$ to obtain the path for the block $S$. For simplicity, we also refer to the set of points $\cup_{j=1}^i S_j$ as a block. Now, we prove the following lemma.

\begin{lemma}\label{lem:blockchain}
A non-crossing Hamiltonian path of $\cup_{j=1}^i S_j$ can be computed for all $1\le i\le t$ that satisfies the two invariants. 
\end{lemma}

\begin{figure}[tbp]
 \centering
 \includegraphics{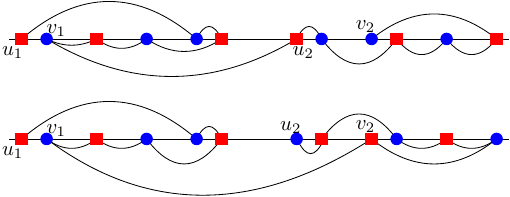}
  \caption{Illustration showing Case 1 (upper) and Case 2 (lower) of Lemma~\ref{lem:blockchain}.}
\label{fig:fig2}
\end{figure}

\begin{proof}
We prove this using induction on the values of $i$. In the base case, for $i=1$, we know how to compute the path of $\cup_{j=1}^i S_j=S_1$ that satisfies the two invariants. Now, consider any $i\ge 2$. Suppose we have already computed the path $T_{i-1}$ of $\cup_{j=1}^{i-1} S_j$ that satisfies the two invariants. Let $\Pi_i$ be the path of $S_i$ that also satisfies the two invariants. Let $u_1,v_1$ (resp. $u_2,v_2$) be the endpoints of the path $T_{i-1}$ (resp. $\Pi_i$) with $u_1 < v_1$ (resp. $u_2 < v_2$). As $\cup_{j=1}^{i-1} S_j$ (resp. $S_i$) contains the same number of red and blue points, the color of $u_1$ (resp. $u_2$) will be different from the color of $v_1$ (resp. $v_2$). Now, there are two cases. 

\paragraph{\emph{\textit{$1.$ $u_1$ and $u_2$ have the same color.}}} We add the edge $(v_1,u_2)$ with $T_{i-1}\cup \Pi_i$ to get the path $T_i$ for $\cup_{j=1}^{i} S_j$ (see Figure \ref{fig:fig2}). To make sure $(v_1,u_2)$ does not cross the other edges, we can make use of the second invariant. From the invariant it follows that, for $v_1$, all the edges of $T_{i-1}$ whose endpoints are on both sides of $v_1$ must lie on the same side of $H$. Thus, if those edges lie above, we draw the edge $(v_1,u_2)$ below $H$. Otherwise, we draw $(v_1,u_2)$ above $H$. Note that $u_1$ is an endpoint of $T_i$ which is the first point of $\cup_{j=1}^{i} S_j$. Also $v_1 < u_2 < v_2$. Thus, the other endpoint $v_2$, still satisfies the second invariant as before by induction. 

\paragraph{\emph{\textit{$2.$ $u_1$ and $u_2$ have different colors.}}} We add the edge $(v_1,v_2)$ with $T_{i-1}\cup \Pi_i$ to get 
the path $T_i$ for $\cup_{j=1}^{i} S_j$ (see Figure \ref{fig:fig2}). We need to ensure that the edges of $T_{i-1}$ and 
$\Pi_i$ whose endpoints are on both sides of $v_1$ and $v_2$, respectively, lie on the same side of $H$. If this is not true, the edges of $\Pi_i$ that lie below $H$ can be redrawn above $H$, and the edges of $\Pi_i$ that lie above $H$ can be redrawn below $H$. This does not violate any invariant. Hence, $(v_1,v_2)$ can be drawn without crossing any edge of $T_{i-1}\cup \Pi_i$. Note that $u_1$ is an endpoint of $T_i$ which is the first point of $\cup_{j=1}^{i} S_j$. The other endpoint $u_2$ was an endpoint of the block $S_i$. Thus, even after the drawing of $(v_1,v_2)$ one side of $u_2$ still remains available. Hence, the second invariant is also satisfied.
\qed
\end{proof}

The next lemma completes the induction step for showing the construction of the path for the level $l$ block. 

\begin{lemma} \label{lem:block}
A non-crossing Hamiltonian path for the level $l$ block $S$ can be computed that satisfies the two invariants. 
\end{lemma}

\begin{figure}[tbp]
\centering
\includegraphics{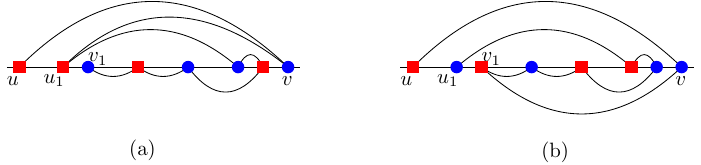}
\caption{(a) Illustration showing Case 1 of Lemma \ref{lem:block}. (b) Illustration showing Case 2 of Lemma \ref{lem:block}.}\label{fig:fig5}
\end{figure}


\begin{proof}
First, we compute the path $T_t$ for the points in $\cup_{j=1}^{t} S_j$ using the construction in Lemma \ref{lem:blockchain}. Let $u_1,v_1$ be the endpoints of $T_t$ such that $u_1<v_1$. Note that as mentioned before $S=(\cup_{j=1}^{t} S_j)\cup \{u,v\}$. Without loss of generality, assume the color of $u$ and $v$ is red and blue, respectively. The other case can be handled similarly. Now, there can be two cases. 

\paragraph{\emph{\textit{$1.$ $u_1$ is red and $v_1$ is blue.}}} We add the edge $(u_1,v)$ with $T_{t}\cup \{(u,v)\}$ to get the path for $S$ (see Figure \ref{fig:fig5}(a)). From the second invariant for $T_t$ it follows that, for $v_1$, all the edges of $T_{t}$ whose endpoints are on both sides of $v_1$ must lie on the same side of $H$. Thus, if those edges lie above, we draw the edges $(u_1,v),(u,v)$ above $H$. Otherwise, we draw $(u_1,v),(u,v)$ below $H$. Thus, the second invariant is satisfied. Also, note that $u$ is an endpoint of the new path which is the first point of $S$. Hence, both the invariants are satisfied.

\paragraph{\emph{\textit{$2.$ $u_1$ is blue and $v_1$ is red.}}} We add the edge $(v_1,v)$ with $T_{t}\cup \{(u,v)\}$ to get the path for $S$ (see Figure \ref{fig:fig5}(b)). From the second invariant for $T_t$ it follows that, for $v_1$, all the edges of $T_{t}$ whose endpoints are on both sides of $v_1$ must lie on the same side of $H$. Thus, if those edges lie above, we draw the edge $(v_1,v)$ below $H$. Otherwise, we draw $(v_1,v)$ above $H$. As $u_1$, an endpoint of the new path, is the second point of $S$, irrespective of how we draw $(u,v)$, the second invariant is satisfied. Also $u$ is an endpoint of the new path which is the first point of $S$. Hence, both the invariants are satisfied in this case as well. \qed
\end{proof}

To compute the path of all the points in $R\cup B$ one can note that $R\cup B$ is the union of a set of blocks having levels at most the maximum level $L$. By Lemma~\ref{lem:block}, we can compute the paths for all such blocks that satisfy the invariants. Then we can merge those paths using the construction in Lemma \ref{lem:blockchain} to get the path for the points in $R\cup B$. It is easy to verify that the overall construction can be done in polynomial time. Thus, we get the following theorem.

\begin{theorem}
For any set $R$ of red points and $B$ of blue points on a line with $|R|=|B|$, there always exists a non-crossing Hamiltonian path whose edges are circular arcs that lie above or below $H$. Moreover, such a path can be computed in polynomial time. 
\end{theorem}

\subsection{A Linear Time Algorithm for Non-Crossing Hamiltonian Path}
In this subsection, we give another algorithm for computing a non-crossing Hamiltonian path. This algorithm uses very different set of ideas than the previous algorithm. Recall that all the input points lie on a line. We assume that the points are given in sorted order with respect to their $x$ coordinates. For a point $p$ (except the last one), let $S(p)$ be the point which is the successor of $p$ in this order. We use the following algorithm to compute a non-crossing Hamiltonian path. In contrast to the previous algorithm, this algorithm processes the points from left to right and extends the Hamiltonian path constructed so far by connecting the current point with an appropriately chosen point. In particular, in every iteration, we consider a point $p$ and connect it by adding one or more edges. Initially, $p$ is the leftmost point. We also maintain a set of active points 
which is initialized to the set of 
all 
points. We store the constructed path in a set of edges $\Pi$, which is initially empty. 

\begin{itemize}
    \item Let $\textsf{Right}(r)$ and $\textsf{Right}(b)$ be the rightmost (or last in the order) red and blue points, respectively, which are active. 
   
    \item If the color of $p$ is different from the color of $S(p)$, we simply add an 
    arc $(p,S(p))$ to $\Pi$ that lies above $H$. Make $p$ inactive. 
    
    \item Otherwise, there are two cases. 
    \begin{enumerate}[(i).]
        \item If $p$ is red, add two edges $(p,\textsf{Right}(b))$ and $(\textsf{Right}(b),S(p))$ to $\Pi$. These two edges are drawn above $H$ as circular arcs. Make $p$ and $\textsf{Right}(b)$ inactive.
        
        \item If $p$ is blue, add two edges $(p,\textsf{Right}(r))$ and $(\textsf{Right}(r),S(p))$ to $\Pi$. These two edges are drawn below $H$ as circular arcs. Make $p$ and $\textsf{Right}(r)$ inactive.
    \end{enumerate}
    
    \item If $S(p)$ is active, assign $S(p)$ to $p$ (i.e., $p\leftarrow S(p)$) and repeat all the steps. Otherwise, terminate the algorithm.  
\end{itemize}

The different iterations of the above algorithm are shown in an example in Figure \ref{fig:lineartsp}. Now we discuss the correctness of the algorithm. First, we have the following observation. 

\begin{figure}[tbp]
 \centering
 \includegraphics{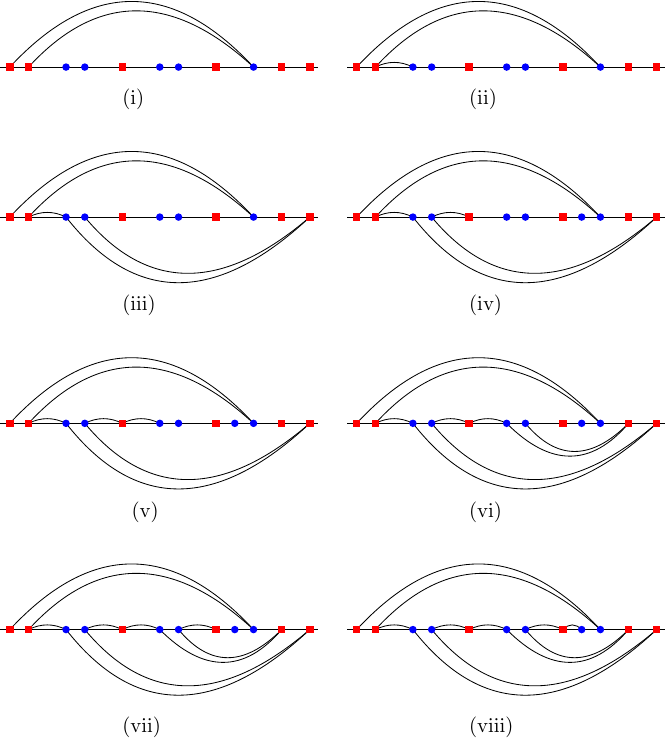}
  \caption{Figure showing the execution of the Hamiltonian path computation algorithm on an example point set.}
\label{fig:lineartsp}
\end{figure}

\begin{obs}\label{obs:markers}
Consider any iteration of the algorithm. Then, any red point on the right of $\textsf{Right}(r)$ (if any) is inactive and has degree 2. Similarly, any blue point on the right of $\textsf{Right}(b)$ (if any) is inactive and has degree 2. Moreover, any point on the left of $p$ (if any) is inactive and except the first point all of them have degree 2.
\end{obs}

\begin{lemma}
The algorithm correctly computes a bichromatic Hamiltonian path. 
\end{lemma}

\begin{proof}
Note that when the algorithm terminates, $S(p)$ is inactive. Thus, its degree must be 2. If $S(p)$ is red (resp. blue), then it had become $\textsf{Right}(r)$ (resp. $\textsf{Right}(b)$) at some point and its degree is 2. By Observation \ref{obs:markers}, all the points whose colors are same as the color of $S(p)$ and lie on the right of $S(p)$ have degree 2. Also, the degree of all the points on the left of $p$ except the first point is 2. The degree of $p$ and the first point is 1. As the number of red and blue points are same, all the points that lie on the right of $S(p)$ must have degree 2. Thus, $\Pi$ is a subgraph where each vertex has degree 2 except two special vertices whose degrees are 1. It follows that such a subgraph is a path that spans all the vertices. Also, all the edges on this path $\Pi$ are bichromatic, and hence $\Pi$ is a valid bichromatic Hamiltonian path. \qed
\end{proof}

Next, we argue that the computed Hamiltonian path is non-crossing. The arcs that are added between points $p$ and $S(p)$ in the second step do not cross any other drawn edges, as $p$ and $S(p)$ are consecutive points. Also, the edges drawn above $H$ do not cross any edges drawn below $H$. Moreover, the edges $(p,\textsf{Right}(r))$ and $(\textsf{Right}(r),S(p))$ (or $(p,\textsf{Right}(b))$ and $(\textsf{Right}(b),S(p))$) drawn in the same iteration do not cross each other. The following observation completes the claim. 

\begin{obs}
Consider two edges $(u,v)$ and $(u',v')$ which are drawn as circular arcs above (resp. below) $H$ and added to $\Pi$ in different iterations. Then, either $(u,v)$ is nested in $(u',v')$, or $(u',v')$ is nested in $(u,v)$. 
\end{obs}

The algorithm can be implemented to run in linear time. Note that given the values of $p$, $\textsf{Right}(r)$ and $\textsf{Right}(b)$, each iteration of the algorithm can be performed in $O(1)$ time. Thus, it is sufficient to show that the number of iterations is at most the number of points. We use three pointers to keep track of $p$, $\textsf{Right}(r)$ and $\textsf{Right}(b)$ in each iteration. In every iteration, we set $S(p)$ to be the new $p$. Thus, the pointer to $p$ always moves from left to right, i.e, it tracks each point at most once. Also, in an iteration, when we change $\textsf{Right}(r)$ (resp. $\textsf{Right}(b)$) the red (resp. blue) point on its left becomes the new $\textsf{Right}(r)$ (resp. $\textsf{Right}(b)$). Thus, the two pointers to $\textsf{Right}(r)$ and $\textsf{Right}(b)$ always move from right to left. Hence, the linear running time of the algorithm follows.    


\begin{theorem}\label{thm:hplinear}
For any set $R$ of red points and $B$ of blue points on a line with $|R|=|B|$, a non-crossing Hamiltonian path can be computed in linear time whose edges are circular arcs that lie above or below $H$. 
\end{theorem}




\section{Minimum Spanning Tree for Collinear Points}
\label{sec:spantree}

In this section, we study the \textsc{Bichromatic spanning tree} problem for collinear points. 
We proceed with the following definition.



\begin{definition}
\textbf{Spanning tree for collinear points}. Given a set $R$ of $n$ red points and a set $B$ of $m$ blue points all of which lie on a line, find a minimum weight geometric tree $T$ in the plane such that each edge of $T$ is represented by a circular arc that lies above $H$, each arc connects a red and a blue point, and $T$ spans all the input points. The weight of an arc is given by the Euclidean distance of its endpoints. In the non-crossing version of the problem, one would like to compute such a tree so that the corresponding circular arcs are non-crossing. 
\end{definition}


First, we discuss a greedy linear time algorithm for computing an optimum, i.e., minimum-weight spanning tree, which potentially has crossings. 

\subsection{Spanning Tree with Crossings}\label{ap:stwc}

Let $p_1,\ldots,p_{n+m}$ be the input points sorted in increasing order of their $x$ coordinates.
For each point $p_i$, let $\col(p_i)$
be the color of $p_i$.
We assume that the points are given in this order. Our algorithm has two steps. In the first step, we traverse the points in the sorted order and connect each point with its nearest opposite color point using an arc if it is not already connected. 
This gives us a set of components $\{C_1,\ldots, C_k\}$ 
for $k\le n$, 
where each component contains at least one edge. 
For any component $C_i, 1 \le  i \le  k$, let $l(C_i)$ and $r(C_i)$ be the leftmost and the rightmost point, respectively. 
In the second step, we traverse the components from left to right. 
Consider the first two components $C_1$ and $C_2$. If $\col(r(C_1))\ne \col(l(C_1))$, join $C_1$ and $C_2$ by an arc $(r(C_1),l(C_2))$.
Otherwise, check the distance between $r(C_1)$ and its nearest opposite color point in $C_2$, and the same for $l(C_2)$ and its nearest opposite color point in $C_1$. Choose the shorter one to join $C_1$ and $C_2$. We repeat the same process for each consecutive pair of the remaining components. See Figure \ref{fig:stwc} for an example. 

\begin{figure}[tbp]
	\centering
	\includegraphics{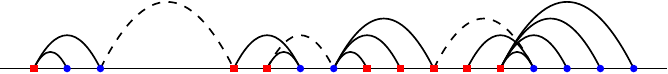}
	\caption{Figure demonstrating the execution of the algorithm on an example. The dashed arcs are added in the second step.}
	\label{fig:stwc}
\end{figure}

Note that after the first step each component $C_i$ is a tree. In the second step, we add exactly one arc between a consecutive pair of components. Hence, the selected arcs form a valid spanning tree. 
Next, we move on towards proving the correctness of the algorithm. We need a few definitions for that. Consider any graph $G=(V,E)$ and a subset of edges $X$ in this graph. Also consider a subset $S$ of vertices. An edge $e$ is said to cross the cut $(S,V\setminus S)$ or be across the cut $(S,V\setminus S)$ if one of its endpoint is in $S$ and another endpoint is in $V\setminus S$. $X$ does not cross the cut $(S,V\setminus S)$ if none of its edges cross the cut.  

To prove the correctness of our algorithm, we need the following standard lemma \cite{DBLP:books/daglib/0017733} from the literature of  
minimum spanning tree. 

\begin{lemma}[Cut property]\label{lem:cutproperty}
\cite{DBLP:books/daglib/0017733} Suppose the set of edges $X$ are part of a minimum spanning tree of $G = (V, E)$. Pick any
subset of vertices $S$ for which $X$ does not cross the cut $(S, V\setminus S)$, and let $e$ be a minimum weight 
edge across this cut. Then $X \cup \{e\}$ is part of some minimum spanning tree. 
\end{lemma}

Next, we prove the optimality of our algorithm. 

\begin{lemma}\label{lem:correct}
At any moment, let $X$ be the set of edges added by the algorithm so far. Then there is an optimum bichromatic spanning tree 
that contains $X$. 
\end{lemma}

\begin{proof}
We prove this lemma by induction. In the base case, $X=\emptyset$ and the lemma is vacuously true. Now, consider any moment when $X$ is the set of edges added so far by the algorithm such that there is an optimum bichromatic spanning tree 
that contains $X$. Let $e=(u,v)$ be the next bichromatic edge added by the algorithm. Now, there can be two cases. $e$ is added in the first or second step. 

Suppose $e$ is added in the first step. We show that $e$ is a minimum weight edge across a cut such that $X$ does not cross the edges of this cut. Then the statement of the lemma holds by the Cut property of Lemma \ref{lem:cutproperty}. Note that in the first step, we connect each point to its nearest opposite color point if it is not already connected. WLOG, suppose $(u,v)$ was added corresponding to $u$. Then $X$ does not cross the cut $(V\setminus \{u\},\{u\})$, as $u$ is not yet connected. Also, the way we connect $u$ to its nearest opposite color point, $(u,v)$ must be a minimum weight edge across $(V\setminus \{u\},\{u\})$ in the (implicit) bichromatic input graph. Thus, the lemma holds in this case.    

Next, suppose $e=(u,v)$ is added in the second step. WLOG, let $u$ be in the component $C_i$ and $v$ be in $C_{i+1}$. 
Consider the cut $(P_1,V\setminus P_1)$ such that $P_1=\cup_{j=1}^i C_j$. Then, by the way we select the edge across $C_i$ and $C_{i+1}$, 
it follows that $(u,v)$ must be a minimum weight (bichromatic) edge across the cut $(P_1,V\setminus P_1)$. Now, consider any edge in $X$. If it is added in the first step, it cannot cross $(P_1,V\setminus P_1)$, as it must lie within a component. Otherwise, the edge is added in the second step. But, in this case its endpoints should lie in $P_1$, as we connect the consecutive components from left to right. By the Cut property, the lemma holds in this case as well.  
\qed
\end{proof}

Now, we prove the key theorem of this subsection. 

\begin{theorem}\label{th-stwc}
For any set $R$ of red points and $B$ of blue points on a line, an optimum spanning tree can be computed in linear time. 
\end{theorem}

\begin{proof}
 To prove this theorem, we run the above algorithm. The correctness follows by Lemma \ref{lem:correct} when the algorithm terminates. As the nearest neighbors (of opposite color) of all the points can be computed and stored in linear time given the sorted points, the algorithm can be executed in linear time.  This completes the proof of the theorem. 
 \qed
\end{proof}

Next, we discuss the algorithm for the non-crossing case.  

\subsection{Non-crossing Spanning Tree}

Let $P_1,P_2,\ldots,P_k$ be the alternating monochromatic chunks of points ordered from left to right for $k\in \{1,\ldots,n\}$. Thus, the color of the points in $P_i$ is different from the color of the points in $P_{i+1}$ for all $1\le i\le k-1$. An arc $(p,q)$ is called a consecutive (resp.  non-consecutive) arc if it connects points from two consecutive (resp.  non-consecutive) chunks. 
We start with the following observation.

\begin{obs}\label{obs:consecutive}
Consider any point $p\in P_i$. If an arc $(p,q)$ is contained in a minimum spanning tree, then either $q\in P_{i-1}$ or $q\in P_{i+1}$, i.e., $(p,q)$ must be a consecutive arc. 
\end{obs}

\begin{proof}
Consider any minimum spanning tree $T$. Suppose there is a non-consecut ive arc in $T$. Consider an arc $(p,q)$ in $T$ such that
$(p,q)$ is a minimal non-consecutive arc, i.e., all arcs nested in $(p,q)$ are consecutive arcs. We claim that it is always possible to replace $(p,q)$ in $T$ by a consecutive arc $(p_1,q_1)$ such that $T_1=(T\setminus \{(p,q)\})\cup \{(p_1,q_1)\}$ is a non-crossing spanning tree and the cost of $T_1$ is strictly less than the cost of $T$. But, this leads to a contradiction, and hence there cannot be any non-consecutive arc in $T$. Next, we prove the claim. 

Let $p'$ be the rightmost point between $p$ and $q$ (including $p$) that is connected to $p$ in $T\setminus \{(p,q)\}$ (see Figure \ref{fig:consecutive}(a)). Let $q'$ be the point on the immediate right of $p'$. Note that $q'$ is connected to $q$ in $T\setminus \{(p,q)\}$. If $p'$ and $q'$ are of different colors, then we replace $(p,q)$ in $T$ by the consecutive arc $(p',q')$ (see Figure \ref{fig:consecutive}(b)). There is no crossings in $T_1=(T\setminus \{(p,q)\})\cup \{(p',q')\}$, as $p'$ and $q'$ are consecutive points. $T_1$ is a spanning tree, as all the points are contained in it and the two connected components in $T\setminus \{(p,q)\}$ are joined in $T_1$ by the arc $(p',q')$. As $(p',q')$ is a consecutive arc nested in the non-consecutive arc $(p,q)$, the cost of $T_1$ is strictly less than the cost of $T$. 
 
Now, suppose $p'$ and $q'$ have the same color. If $q\ne q'$, consider any arc $e_1$ in $T\setminus \{(p,q)\}$ which contains $q'$ as an endpoint and there is no other arc in $T\setminus \{(p,q)\}$ containing $q'$ which nests $e$. Let $e_1=(q',q'')$ (see Figure \ref{fig:consecutive}(c)). The color of $p'$ and $q''$ are different, as $e_1$ is bichromatic and $p'$ and $q'$ have the same color. Note that if $q=q'$, there might not be any arc in $T\setminus \{(p,q)\}$ containing $q'$, and thus $e_1$ cannot be defined. Similarly, if $p\ne p'$, consider any arc $e_2$ in $T\setminus \{(p,q)\}$ which contains $p'$ as an endpoint and there is no other arc in $T\setminus \{(p,q)\}$ containing $p'$ which nests $e_2$. Let $e_2=(p',p'')$ (see Figure \ref{fig:consecutive}(c)). The color of $q'$ and $p''$ are different, as $e_2$ is bichromatic and $p'$ and $q'$ have the same color. Note that if $p=p'$, there might not be any arc in $T\setminus \{(p,q)\}$ containing $p'$, and thus $e_2$ cannot be defined. But, at least one of $e_1$ or $e_2$ must be defined, as $p=p'$ and $q=q'$ cannot happen at the same time. Otherwise, $(p,q)$ becomes consecutive. Wlog, suppose $e_1$ exists. In this case, we replace $(p,q)$ in $T$ by the consecutive arc $(p',q'')$ (see Figure \ref{fig:consecutive}(d)). There is no crossings in $T'=(T\setminus \{(p,q)\})\cup \{(p',q'')\}$, as $(p',q'')$ nests $(q',q'')$ and does not cross any arc in $T$ nested in $(p,q)$. $T'$ is a spanning tree, as all the points are contained in it and the two connected components in $T\setminus \{(p,q)\}$ are joined in $T'$ by the arc $(p',q'')$. As $(p',q'')$ is a consecutive arc nested in the non-consecutive arc $(p,q)$, the cost of $T_1$ is strictly less than the cost of $T$. This concludes the proof of the claim and so the proof of the observation.\qed
 
\end{proof}

\begin{figure}[tbp]
	\centering
	\includegraphics{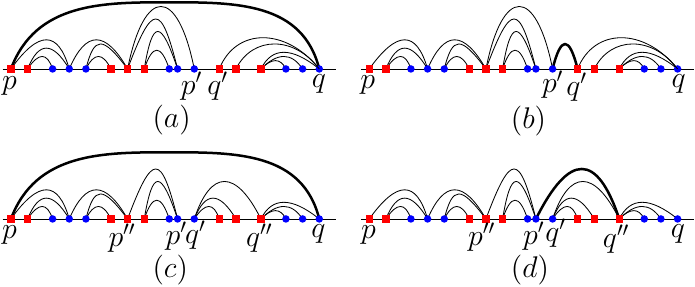}
	\caption{(a),(b) show the case when $p'$ and $q'$ are of different colors. (c),(d) show the case when $p'$ and $q'$ have same color.  (b) shows the tree obtained after replacing $(p,q)$ in (a) by $(p',q')$. (d) shows the tree obtained after replacing $(p,q)$ in (c) by $(p',q'')$. }
	\label{fig:consecutive}
\end{figure}


As the spanning tree we want to compute is non-crossing, by the above observation, it follows that all the arcs between two consecutive chunks are nested. 

\begin{obs}
Consider any two arcs $(p_1,q_1)$ and $(p_2,q_2)$ in a minimum non-crossing spanning tree such that $p_1,p_2\in P_i$ and $q_1,q_2\in P_{i+1}$. Then, either $(p_2,q_2)$ is nested in $(p_1,q_1)$ or $(p_1,q_1)$ is nested in $(p_2,q_2)$. 
\end{obs}

\begin{proof}
 This observation follows from the fact that if none of the two arcs is nested in the other, then they must cross. \qed
\end{proof}

\vspace{-.1cm}
The above observation implies that the outermost arcs between consecutive chunks form a path (an umbrella) between the first and the last point and all the other arcs lie inside this umbrella (see Figure \ref{fig:DP-split}). 
Next, we give a simple algorithm to compute an optimum spanning tree inside such an outermost arc. Suppose $p_0,p_1,\ldots,p_l,\ldots,p_{\tau+1}$ are points in sorted order such that $\{p_0,p_1,\ldots,p_l\}\subseteq P_i$ and $\{p_{l+1},\ldots,p_{\tau+1}\}\subseteq P_{i+1}$. We would like to construct an optimum spanning tree of the points $p_0,p_1,\ldots,p_l,\ldots,p_{\tau+1}$ which contains the arc $(p_0,p_{\tau+1})$. Our algorithm is based on the following observation. 

\begin{figure}[tbp]
	\centering
	\includegraphics{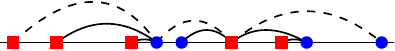}
	\caption{Figure showing a spanning tree with the umbrella shown by dashed arcs.}
	\label{fig:DP-split}
\end{figure}

\begin{obs}
Any optimum spanning tree that contains $(p_0,p_{\tau+1})$ must also contain either $(p_0,p_{\tau})$ or $(p_1,p_{\tau+1})$ whichever has lower weight. 
\end{obs}

\begin{proof}
 Consider any optimum non-crossing spanning tree $T$ that contains the arc $(p_0,p_{\tau+1})$. Also, suppose $T$ does not contain a cheaper arc from  $\{(p_0,p_{\tau}), (p_1,$ $p_{\tau+1})\}$. Otherwise, we are done with the proof. If $T$ does not contain any arc from $\{(p_0,p_{\tau}), (p_1,p_{\tau+1})\}$, then at least one of $p_1$ or $p_{\tau}$ cannot be connected to $T$ without introducing any crossings. It follows that $T$ contains the arc $e$ from $\{(p_0,p_{\tau}), (p_1,p_{\tau+1})\}$ having the larger weight. But, in that case one can replace $e$ in $T$ by the cheaper arc $e'$. $T' = (T\setminus \{e\})\cup \{e'\}$ is a non-crossing spanning tree, as the only arc in $T$ that $e'$ crosses is $e$ which is now removed. But, then the weight of $T'$ is strictly lesser than that of $T$, which is a contradiction. Hence, the observation follows. \qed
\end{proof}

In our algorithm to compute an optimum spanning tree inside an outermost arc $(p_0,p_{\tau+1})$, first we select the shorter arc among $(p_0,p_{\tau})$ and $(p_1,p_{\tau+1})$. Then, we recursively solve the problem inside the selected arc by treating it as an outermost arc. This problem can be solved in linear time. We are going to use this algorithm as a subroutine in our algorithm for solving the general problem. Next, we give a Dynamic Programming (DP) based algorithm for this general problem. This algorithm essentially  decides  which outermost arcs to choose. 

Let $p_1,p_2,\ldots, p_{n+m}$ 
be the input points. Our DP algorithm incrementally computes a non-crossing spanning tree starting from the left connecting a new point in each step. Let $P_1=\{p_1,\ldots,p_l\}$ and $P_2=\{p_{l+1},\ldots,p_{\tau}\}$. For $l+1\le i\le n+m$, let MST$(i)$ be the subproblem of computing an optimum spanning tree for the prefix set of points $\{p_1,\ldots,p_i\}$. We store the cost of MST$(i)$ in table $M$ indexed by $i$, where $l+1\le i\le n+m$. To initialize, for each $l+1\le i\le \tau$, we compute the cost of an optimum spanning tree of $\{p_1,\ldots,p_i\}$. Note that in this base case, any non-crossing spanning tree contains the outermost arc $(p_1,p_i)$. Otherwise, either $p_1$ or $p_i$ cannot be connected without crossing. Thus, we can compute an optimum spanning tree using the algorithm mentioned above for computing an optimum spanning tree inside an outermost arc. Thus, in the base case we compute all the entries of $M$ with indices in $\{l+1,\ldots,\tau\}$. 

Now, suppose we want to solve MST$(i)$ for any $i\ge \tau+1$. Thus, we would like to connect a new point $p_i\in P_{t+1}$ for some $t\ge 2$. We have already computed the entry $M[q]$ for all $l+1\le q < i$. 
By Observation \ref{obs:consecutive}, in the solution spanning tree of MST$(i)$, $p_i$ must be connected to a point $p_s$ of $P_t$, as $p_i\in P_{t+1}$. For each such $p_s$, we compute the cost of the spanning tree that contains the arc $(p_s,p_i)$. In particular, the total cost is the sum of three costs: (i) the cost of the arc $(p_s,p_i)$, (ii) the cost of connecting the points inside the open interval $(p_s,p_i)$ and (iii) the cost of the optimum spanning tree of $\{p_1,\ldots,p_s\}$. 
We select the point $p_s$ that minimizes the total cost. Thus, $$M[i]=\min_{s:p_s\in P_t} \lVert p_ip_s\rVert+\text{cost}(s,i)+M[s].$$ Here, $\lVert p_ip_s\rVert$ is the cost of $(p_s,p_i)$ and \text{cost}$(s,i)$ is the cost of any optimum spanning tree inside the outermost arc $(p_s,p_i)$. Note that one can precompute and store  the costs  \text{cost}$(s,i)$ for all possible open intervals $(p_s,p_i)$ in quadratic time using the algorithm for computing an optimum spanning tree inside an outermost arc. 
Thus, for a fixed $p_s$, the sum of the three costs mentioned above can be computed using a constant number of table lookups. Thus, each step of this dynamic programming based algorithm takes linear time. Hence, an optimum spanning tree can be computed in quadratic time. 

\begin{theorem}\label{th-ncst}
For any set $R$ of red points and $B$ of blue points on a line, an optimum non-crossing spanning tree can be computed in quadratic time. 
\end{theorem}

\section{Minimum Non-crossing Matching for Collinear Points}
\label{sec:matching}

Note that the fact that a minimum-weight bichromatic matching for points in general position is always non-crossing 
might not hold in the case of collinear points. Indeed, there are point sets for which no non-crossing matching exists if the edges are represented by segments. However, one can show that there is always a non-crossing matching of collinear points such that each matching edge is a circular arc drawn above the line. 
Again the weight of an arc is the Euclidean distance between its endpoints. We say that two arcs are \textit{pairwise disjoint} if their endpoints are disjoint. 

\begin{definition}
\textbf{Non-crossing matching for collinear points}. Given a set $R$ of $n$ red points and a set $B$ of $n$ blue points all of which lie on a line, find a set of $n$ pairwise disjoint and non-crossing circular arcs in the plane of minimum total weight such that the arcs lie above $H$, each arc connects a red and a blue point, and the arcs span all the input points. 
\end{definition}

Using the bipartite matching algorithm due to Kaplan~{\em et al.}~\cite{KaplanMRSS17} along with a simple postprocessing (already described in the introduction), one can immediately solve this problem in $O(n^2 \poly(\log n))$ time. 
Here we design a simple algorithm with improved $O(n)$ time complexity. 

Let $p_1,p_2,\ldots,p_{2n}$ be the input points sorted from left to right based on their $x$ coordinates. We assume that the points are given in this order. 
For any point $p_i\in P$, let $\col(p_i)$ denote the color of $p_i$. A subset of points $P_i\subseteq P$ is called \emph{color-balanced} if it contains an equal number of red and blue points. 
We traverse the points from left to right and seek  
the first balanced subset (denoted by $P_1$).
In order to obtain $P_1$ we use a simple method. 
We start with the leftmost point $p_1$ and maintain a counter $C$ which is used to find the balanced subset and is initialized to $0$ at the beginning. 
If $\col(p_1)=$ red, we increase the value of $C$ by $1$,
and decrease by $1$, otherwise. 
Observe that we will get a balanced subset when the value of $C$ becomes $0$.
Let $P_1\subseteq P$ be the first balanced subset containing 
$2k$ (for some $k\in \{1,\ldots,n\}$) points.
The remaining points $P\setminus P_1$ also form a balanced subset since $P$ contains exactly $n$ red and $n$ blue points. We prove the following lemma.


\begin{lemma}
\label{matching-lemma}
Let $P_1\subseteq P$ be the first color-balanced subset of $P$ and $|P_1|=2k$. Then $\col(p_1)\ne \col(p_{2k})$, and any minimum non-crossing perfect matching $M_P$ of $P$ contains the edge $(p_1,p_{2k})$.
\end{lemma}

\begin{proof}
The first part of the lemma is clearly true, otherwise 
the value of the counter would not be $0$ at $p_{2k}$, which is the termination criteria to obtain the first balanced subset. Now, let us assume that $M_P$ does not contain the edge $(p_1,p_{2k})$.
Then one of the following two situations can happen: 1) $p_1$ and $p_{2k}$ are matched with two intermediate points from $P_1$;
2) one or both of $p_1$ and $p_{2k}$ are matched with points from $P\setminus P_1$. 

\paragraph{Case~$1$:} $p_1$ and $p_{2k}$ are matched with two intermediate points $p_\tau$ and $p_{\ell}$, respectively. 
Note $\ell>\tau$, otherwise the matching edges cross each other. We know that $\{p_1,\ldots,p_\tau\}$ is not a balanced subset since $P_1$ is the first balanced subset. Therefore, there exists at least one point $p_r$ (where $1<r<\tau$) that is matched with a point $p_s$ (where $s>\tau$). In that case, the edge $(p_r,p_s)$ will intersect $(p_1,p_\tau)$. Hence, we get a contradiction. 

\paragraph{Case~$2$:} Suppose both of $p_1$ and $p_{2k}$ are matched with points from $P\setminus P_1$ and no other point from $\{p_2,\ldots,p_{2k-1}\}$ is matched with any point from $P\setminus P_1$. Then we can construct a new matching by adding the edge $(p_1,p_{2k})$ and by matching the two points in $P\setminus P_1$. The new matching has lesser cost and is non-crossing; see Figure~\ref{fig:matching-case-2-1}(a). 
If any other point in $\{p_2,\ldots,p_{2k-1}\}$ (say $p_x$) is also matched with a point in $P\setminus P_1$, then we know it must be of opposite color of either $p_1$ or $p_{2k}$, since $\col(p_1)\ne \col(p_{2k})$. Hence, we can either add the edge $(p_1,p_x)$ or $(p_x,p_{2k})$ and this reduces the total cost; see Figure~\ref{fig:matching-case-2-1}(b). 
The new matching might not be non-crossing. But, using similar argument one can remove all the crossings without increasing the cost. Thus, at the end we get a cheaper non-crossing matching, which contradicts the optimality of $M_P$.

\begin{figure}[tbp]
\centering
\includegraphics{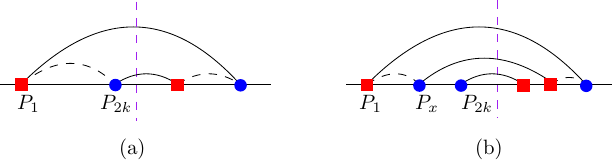}
\caption{Figure demonstrating the two situations in the case when both $p_1$ and $p_{2k}$ are matched with points from $P\setminus P_1$.}
\label{fig:matching-case-2-1}
\end{figure}

Now, if only one of $p_1$ or $p_{2k}$ is matched with a point from $P\setminus P_1$, let us assume $p_1$, then we know there must be at least one other point (say $p_x\in P_1$) that is also matched with a point from $P\setminus P_1$, and $\col(p_1)\ne \col(p_x)$. We can apply similar arguments as above to get a contradiction, which concludes the proof of the lemma.
\qed
\end{proof}

Now, we use Lemma~\ref{matching-lemma} to proceed with the algorithm. 
First, we obtain the balanced subset $P_1$, 
and match the points $p_1$ and $p_{2k}$ by an arc and include the edge $(p_1,p_{2k})$ in $M_P$. 
This edge partitions the point set into two color-balanced subsets, i.e.,
$P_2=P\setminus P_1$ and $P'_1=P_1\setminus \{p_1,p_{2k}\}$.
On each of these subsets we recursively perform the same procedure. 
This process is repeated until each point of $P$ is matched. 

Due to Lemma~\ref{matching-lemma}, we know that every edge we choose in our algorithm must be part of the optimum solution, and no two edges cross each other. Next we show how to convert our recursive algorithm to a non-recursive one, in order to implement it in linear time.  
Consider the points in left to right order, and insert the leftmost point $p_1$ onto a stack. Now, if the next point $p_2$ is of same color as the stack top $p_1$, then push $p_2$ onto the stack; otherwise, match $p_1,p_2$ and remove $p_1$ from the stack. Repeat this process until all points are considered. Indeed, this algorithm is same as the algorithm for matching of parentheses. 

The above linear time algorithm has the same effect as the recursive algorithm. The only difference is that like before when $p_1$ is matched with $p_{2k}$, now all the points in between $p_1$ and $p_{2k}$ are already matched. This can be proved by induction, as $P'_1=P_1\setminus \{p_1,p_{2k}\}$ is also a color-balanced set. We conclude with the following theorem. 

\begin{theorem}\label{th-matching}
For any set $R$ of red points and $B$ of blue points on a line with $|R|=|B|$, an optimum non-crossing matching can be computed in linear time. 
\end{theorem}

\section{Conclusion and Open Problems}
In this paper, we have studied three classical graph problems on geometric graphs induced by bichromatic points for collinear points. We have shown that almost all of these problems can be solved in linear time in this setting. 
We note that the results for circular-arc edges trivially extend to other types of edges drawn in a topologically equivalent way within the same halfplane, e.g, 1-bend polylines.
One problem that is left open by our work is 
the complexity of Bichromatic TSP for collinear points. 
We have shown that if a Hamiltonian path is forced to lie above the line, then such a path might not exist. One  interesting question is to design an algorithm for finding a path in this setting if one exists.  
Another interesting open problem is to design a subquadratic time algorithm for computing a non-crossing spanning tree. 



\old{
\section{TS Tour for Chunked Points on a Circle}
\label{sec:btsp}

Finally, we study the Traveling Salesman problem on the following special point configuration on a circle.

\begin{definition}
\textbf{TS tour for chunked points on a circle}. We are given a set of $n$ red points and a set of $n$ blue points all of which lie on a fixed circle. All points are distributed equidistantly on the circle. Further, the input points are divided into alternately-colored \emph{chunks}, where each chunk contains exactly $k$ consecutive points of the same color. 
The goal is to find a geometric (closed) tour $\pi$ in the plane of minimum total length such that $\pi$ consists of segments each of which connects a red and a blue point and $\pi$ spans all the input points.
\end{definition}

\begin{figure}[tbp]
\centering
\includegraphics{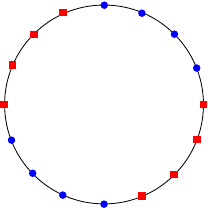}
\caption{Figure demonstrating the point configuration we consider for TS tour, where $n=8$ and $k=4$.}
\label{fig:TSPpoints}
\end{figure}

Note that by definition, $n/k$ is an integer. The total number of chunks is $L=2n/k$ of which $n/k$ contain only red points and $n/k$ contain only blue points. The configuration of the points mentioned in the above definition is shown in Figure \ref{fig:TSPpoints} with an example. 
For any arc between two points $u$ and $v$ on the circle $C$, we denote the arc by $c(uv)$ (resp. $a(uv)$) if it is the clock (resp. anticlock) -wise traversal from $u$ to $v$ along $C$. As any two consecutive points are a fixed distance apart we measure the length of any bichromatic arc (a straight line segment) $uv$ by the minimum of the number of points on the arcs $c(uv)$ and $a(uv)$, respectively (including $u$ and $v$). Next, we design an algorithm for computing a TS tour for the input points. We consider two cases: (i) $k$ is even and (ii) $k$ is odd. 

\paragraph{\textbf{\emph{$(i)$ $k$ is Even}}.} The algorithm in this case is as follows.

\paragraph{\emph{1.}} Let $2p=k$. Partition each chunk into two subchunks each containing $p$ consecutive points. Merge all consecutive subchunks of different colors to form groups. Note that each group contains $p$ red and $p$ blue points. We still preserve the geometry of the points of each group and identify the two peripheral red and blue points of each group as special points. We first compute a bichromatic path between the two special points for each group and later connect the special points of different groups to construct a TS tour for all the points. 

\paragraph{\emph{2.}} For each group, we compute the TS tour in the following way. Consider the ordering of the groups w.r.t clockwise traversal of the points and consider the $i^{th}$ group in this order. WLOG, assume that the red points are visited before the blue points while traversing the points of the group in clockwise order. Let $r_1^i,r_2^i,\ldots,r_{p}^i,b_p^i,b_{p-1}^i,\ldots,b_1^i$ be the points in this order. Join $r_p^i$ with $b_p^i$ and $b_{p-1}^i$ using two arcs. For each $p-1\ge j\ge 2$, join $r_j^i$ with $b_{j+1}^i$ and $b_{j-1}^i$. Finally, join $r_1^i$ with $b_2^i$ (see Figure \ref{fig:even}). Note that each of the points in the group except $r_1^i$ and $b_1^i$ is connected to two points. $r_1^i$ and $b_1^i$ are connected to only one point. 
 
\paragraph{\emph{3.}} Next, we connect the special points of different groups. Recall that $L$ is the total number of groups. Let for the first group the red points are visited before the blue points while traversing the points of the group in clockwise order. Note that the special points are $r_1^1,b_1^1,b_1^2,r_1^2,r_1^3,b_1^3,\ldots,b_1^L,r_1^L$. For $1\leq i\leq L-1$, we connect $r_1^i$ to $b_1^{i+1}$. We also connect $r_1^L$ to $b_1^1$ (see Figure \ref{fig:even}). 

\begin{figure}[tbp]
	\centering
	\includegraphics{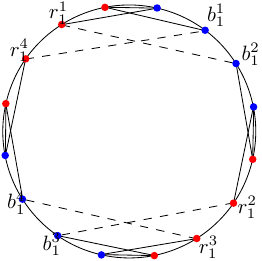}
	\caption{An example of the tour computation for the even case. The arcs between points of same (resp. different) groups are shown using regular (resp. dashed) segments.}
	\label{fig:even}
\end{figure}

It is not hard to see that the set of selected arcs form a valid traveling salesman tour. This is because each of the points is connected to exactly two other points of opposite color. Now, we give a bound on the length of the tour. 

\begin{lemma} \label{lem:tspcost}
 \label{lem:eventsplen} The length of the computed tour is $n(k+2+2/k)$.
\end{lemma}

\begin{proof}
 Note that we have added two types of arcs (1) between points of same group, and (2) between the special points. The length of an arc of the second type is exactly $k+1$, and the number of those arcs is $L$. For each group, we have $k-1$ type 1 arcs whose total length is, 
 \begin{align*}
  & 2+(3+3)+(5+5)+\ldots+(k-1+k-1)\\
  = & 2\cdot (1+3+\ldots+k-1)\\
  = & 2\cdot (k^2/4) = k^2/2
 \end{align*}
Thus, for all the groups, the total length of the type 1 arcs is $Lk^2/2$. Hence, the total length of all arcs is $L((k^2 +2k+2)/2)=n(k+2+2/k)$. \qed
\end{proof}

The algorithm for the odd case is as follows.

\begin{enumerate}
 \item Let $2p+1=k$. Partition each chunk without the middle point into two subchunks each containing $p-1$ consecutive points. Add the middle point to both subchunks. Merge all consecutive subchunks of different colors to form groups. Note that each group contains $p$ red and $p$ blue points. We still preserve the geometry of the points of each group and identify the two peripheral red and blue points as special points. For each group, We compute a bichromatic path between the two special points. As each pair of consecutive groups share a special point, the union of the paths corresponding to the groups form a valid tour.
 \item For each group, we compute the TS tour exactly in the same way as in the $k$ is even case. 
\end{enumerate}

\begin{figure}[tbp]
	\centering
	\includegraphics{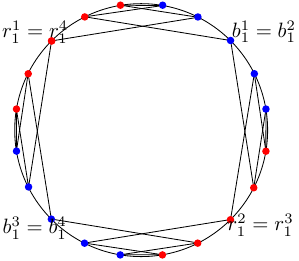}
	\caption{An example of the tour computation for the odd case.}
	\label{fig:odd}
\end{figure}

Now, we give a bound on the length of the tour. 

\begin{lemma}
 \label{lem:oddtsplen} The length of the computed tour is $n(k+2+1/k)$.
\end{lemma}

\begin{proof}
 For each group, we have $k$ arcs whose total length is, 
 \begin{align*}
  & 2+(3+3)+(5+5)+\ldots+(k+k)\\
  = & 2\cdot (1+3+\ldots+k)\\
  = & (k+1)^2/2
 \end{align*}
Thus, for all the groups, the total length of the arcs is $(2n/k)\cdot ((k^2+2k+1)/2)=n(k+2+1/k)$. \qed
\end{proof}
Note that the algorithm can be easily implemented in linear time assuming the points are given in clockwise or anticlockwise order. 

\subsection{Lower Bound}

\begin{lemma}
\label{lem:tsplb}
 The length of any bichromatic traveling salesman tour for the configuration of the points on the circle is at least $n(k+2+2/k)$ if $k$ is even and at least $n(k+2+1/k)$ if $k$ is odd.
\end{lemma}

\begin{sproof}
 Assume $k$ is even. The odd case can be handled in a similar way. We will show that the weight of the arcs adjacent to each red chunk is at least $k^2+2k+2$ in any tour. As the arcs adjacent to two distinct red chunks are disjoint, and the number of red chunks is $n/k$ the lemma follows. Consider a fixed red chunk. The $k$ red points in the chunk must be connected to blue points with $2k$ distinct arcs. WLOG, we assume that the only blue points with which these $k$ points are connected are the $2k$ blue points of the two adjacent chunks. Partition the red chunk into two subchunks each containing $k/2$ consecutive points. Again WLOG, we can assume that the points of a subchunk are only connected to the blue points of the adjacent chunk. Now, consider a fixed subchunk. Note that the $k/2$ points should have a total of degree $k$ which will come from the $k$ blue points. Also each blue point can be connected to only two red points. If the degree $k$ come from only $k/2$ blue points, then the subtour involving the $k/2$ red and $k/2$ blue points form a cycle. Hence, the $k/2$ red points connect to at least $k/2+1$ blue points. Also it is beneficial to connect the red points to the closest $k/2+1$ blue points among the points of the blue chunk. Moreover, to satisfy a total degree $k$ from these $k/2+1$ blue points, it is always beneficial to have only one degree from the farthest and the second farthest blue points, and two degree from each of the remaining $k/2-1$ points. 
 
 Now, there cannot be the case that the first red point $r_1$ (adjacent to the blue chunk) gets connected to the farthest and the second farthest blue points. Otherwise, the remaining $k/2-1$ red and $k/2-1$ blue points will form a cycle. Thus, at most one among the farthest and the second farthest blue points gets connected to $r_1$. The cost to connect the farthest blue point is at least $k/2+2$. The arcs that connect the remaining $k/2$ blue points and $k/2$ red points form a path between the second farthest blue point and a red point. The idea is to prove that the cost to connect these points is at least $(k^2+k-2)/2$. Thus, the total cost is at least $(k^2+2k+2)/2$ for one red subchunk. For both red subchunks the cost is $k^2+2k+2$, and thus the lemma follows. \qed

\end{sproof}

\begin{theorem}
For any set $R$ of red points and $B$ of blue points on a circle with $|R|=|B|$, an optimum non-crossing TS tour can be computed in linear time. 
\end{theorem}

\section{Conclusion and Open Problems}
In this paper, we have studied three classical graph problems on geometric graphs induced by bichromatic points in two restricted settings: (i) collinear points and (ii) equidistant points on a circle. We have shown that almost all of these problems can be solved in linear time in these settings. For the case of collinear points, the results are obtained for graphs whose edges can be drawn as circular arcs. 
We note that the results for collinear points and circular-arc edges trivially extend to other types of edges drawn topologically equivalent within the same halfplane, e.g, 1-bend polylines.
One problem that is left open by our work is to decide the complexity of Bichromatic TSP for collinear points. 

}

\bibliographystyle{abbrv}
\bibliography{main.bib}

\begin{thebibliography}{10}

\bibitem{Abu-AffashBCMS15}
A.~K. Abu{-}Affash, A.~Biniaz, P.~Carmi, A.~Maheshwari, and M.~H.~M. Smid.
\newblock Approximating the bottleneck plane perfect matching of a point set.
\newblock {\em Comput. Geom.}, 48(9):718--731, 2015.

\bibitem{AgarwalES99}
P.~K. Agarwal, A.~Efrat, and M.~Sharir.
\newblock Vertical decomposition of shallow levels in 3-dimensional
  arrangements and its applications.
\newblock {\em {SIAM} J. Comput.}, 29(3):912--953, 1999.

\bibitem{aichholzer2019hamiltonian}
O.~Aichholzer, C.~Alegr{\'\i}a, I.~Parada, A.~Pilz, J.~Tejel, C.~D. T{\'o}th,
  J.~Urrutia, and B.~Vogtenhuber.
\newblock Hamiltonian meander paths and cycles on bichromatic point sets.
\newblock In {\em XVIII Spanish Meeting on Computational Geometry}, pages
  35--38, 2019.

\bibitem{anily1992swapping}
S.~Anily and R.~Hassin.
\newblock The swapping problem.
\newblock {\em Networks}, 22(4):419--433, 1992.

\bibitem{Arora98}
S.~Arora.
\newblock Polynomial time approximation schemes for {E}uclidean traveling
  salesman and other geometric problems.
\newblock {\em J. {ACM}}, 45(5):753--782, 1998.

\bibitem{bk-btg-79}
F.~Bernhart and P.~C. Kainen.
\newblock The book thickness of a graph.
\newblock {\em J. Combinatorial Theory, Series B}, 27(3):320--331, 1979.

\bibitem{biniaz2018spanning}
A.~Biniaz, P.~Bose, D.~Eppstein, A.~Maheshwari, P.~Morin, and M.~Smid.
\newblock Spanning trees in multipartite geometric graphs.
\newblock {\em Algorithmica}, 80(11):3177--3191, 2018.

\bibitem{BiniazBMS16}
A.~Biniaz, P.~Bose, A.~Maheshwari, and M.~H.~M. Smid.
\newblock Plane geodesic spanning trees, {H}amiltonian cycles, and perfect
  matchings in a simple polygon.
\newblock {\em Comput. Geom.}, 57:27--39, 2016.

\bibitem{BiniazBMS18}
A.~Biniaz, P.~Bose, A.~Maheshwari, and M.~H.~M. Smid.
\newblock Plane bichromatic trees of low degree.
\newblock {\em Discrete {\&} Computational Geometry}, 59(4):864--885, 2018.

\bibitem{BorgeltKLLMSV09}
M.~G. Borgelt, M.~J. van Kreveld, M.~L{\"{o}}ffler, J.~Luo, D.~Merrick, R.~I.
  Silveira, and M.~Vahedi.
\newblock Planar bichromatic minimum spanning trees.
\newblock {\em J. Discrete Algorithms}, 7(4):469--478, 2009.

\bibitem{ColanninoDHLMRST07}
J.~Colannino, M.~Damian, F.~Hurtado, S.~Langerman, H.~Meijer, S.~Ramaswami,
  D.~L. Souvaine, and G.~Toussaint.
\newblock Efficient many-to-many point matching in one dimension.
\newblock {\em Graphs Comb.}, 23(Supplement-1):169--178, 2007.

\bibitem{DBLP:books/daglib/0017733}
S.~Dasgupta, C.~H. Papadimitriou, and U.~V. Vazirani.
\newblock {\em Algorithms}.
\newblock McGraw-Hill, 2008.

\bibitem{EvansLMW16}
W.~S. Evans, G.~Liotta, H.~Meijer, and S.~K. Wismath.
\newblock Alternating paths and cycles of minimum length.
\newblock {\em Comput. Geom.}, 58:124--135, 2016.

\bibitem{frank1998bipartite}
A.~Frank, E.~Triesch, B.~Korte, and J.~Vygen.
\newblock On the bipartite travelling salesman problem.
\newblock 1998.

\bibitem{garcia2017polynomially}
A.~Garc{\'\i}a and J.~Tejel.
\newblock Polynomially solvable cases of the bipartite traveling salesman
  problem.
\newblock {\em European Journal of Operational Research}, 257(2):429--438,
  2017.

\bibitem{KanekoK99}
A.~Kaneko and M.~Kano.
\newblock Straight-line embeddings of two rooted trees in the plane.
\newblock {\em Discrete {\&} Computational Geometry}, 21(4):603--613, 1999.

\bibitem{kaneko2003discrete}
A.~Kaneko and M.~Kano.
\newblock Discrete geometry on red and blue points in the plane—a survey—.
\newblock In {\em Discrete and computational geometry}, pages 551--570.
  Springer, 2003.

\bibitem{kaneko1}
A.~Kaneko, M.~Kano, and K.~Suzuki.
\newblock Balanced partitions and path covering of two sets of points in the
  plane.
\newblock {\em preprint}.

\bibitem{KaplanMRSS17}
H.~Kaplan, W.~Mulzer, L.~Roditty, P.~Seiferth, and M.~Sharir.
\newblock Dynamic planar {V}oronoi diagrams for general distance functions and
  their algorithmic applications.
\newblock In {\em Symposium on Discrete Algorithms, SODA 2017}, pages
  2495--2504, 2017.

\bibitem{KynclPT08}
J.~Kyncl, J.~Pach, and G.~T{\'{o}}th.
\newblock Long alternating paths in bicolored point sets.
\newblock {\em Discrete Mathematics}, 308(19):4315--4321, 2008.

\bibitem{MulzerV20}
W.~Mulzer and P.~Valtr.
\newblock Long alternating paths exist.
\newblock In S.~Cabello and D.~Z. Chen, editors, {\em 36th International
  Symposium on Computational Geometry, SoCG 2020, June 23-26, 2020,
  Z{\"{u}}rich, Switzerland}, volume 164 of {\em LIPIcs}, pages 57:1--57:16.
  Schloss Dagstuhl - Leibniz-Zentrum f{\"{u}}r Informatik, 2020.

\end{thebibliography}

\end{document}